\documentclass{sig-alternate}

\usepackage{amsmath}
\usepackage{amssymb}
\usepackage{algorithm}
\usepackage{upgreek}
\usepackage{graphicx}
\usepackage{float}
\usepackage[noend]{algpseudocode}
\usepackage{verbatim}
\usepackage{pbox}
\usepackage[hyphens]{url}
\usepackage{hyperref}
\usepackage{breakurl}

\newtheorem{theorem}{Theorem}

\newtheorem{definition}[theorem]{Definition}

\newcommand{\pFq}[2]{\ensuremath{\!\,_{#1}F_{#2}}}
\newcommand{\bigpFq}[5]{\ensuremath{\pFq{#1}{#2}\left(\begin{matrix}#3 \\ #4 \end{matrix} \;\; \vline  \;\; #5\right)}}

\newcommand   \OO     {O}
\newcommand   \OOsoft {\tilde O}
\newcommand   \M      {\mathsf{M}}

\newcommand   \MZ     {\mathsf{M}_{\mathbb{Z}}}

\usepackage{color}

\definecolor{dblackcolor}{rgb}{0.0,0.0,0.0}
\definecolor{dbluecolor}{rgb}{0.01,0.02,0.7}
\definecolor{dgreencolor}{rgb}{0.2,0.4,0.0}
\definecolor{dgraycolor}{rgb}{0.30,0.3,0.30}

\algrenewcommand\algorithmicrequire{\textbf{Precondition:}}
\algrenewcommand\algorithmicensure{\textbf{Postcondition:}}

\renewcommand{\algorithmicrequire}{\textbf{Input:}}
\renewcommand{\algorithmicensure}{\textbf{Output:}}

\begin{document}

\title{Evaluating parametric holonomic sequences using rectangular splitting}

\numberofauthors{1}

\author{%
 \alignauthor Fredrik Johansson\titlenote{Supported by the Austrian Science Fund (FWF) grant Y464-N18.}\\[\medskipamount]
      \affaddr{\strut RISC}\\
      \affaddr{\strut Johannes Kepler University}\\
      \affaddr{\strut 4040 Linz, Austria}\\[\smallskipamount]
      \email{\strut fredrik.johansson@risc.jku.at}
}

\maketitle

\begin{abstract}
We adapt the rectangular splitting technique of
Paterson and Stockmeyer to the problem of evaluating
terms in holonomic sequences that depend on a parameter.
This approach allows computing the $n$-th term
in a recurrent sequence of suitable type using $O(n^{1/2})$ ``expensive'' operations
at the cost of an increased number of ``cheap'' operations.

Rectangular splitting has little overhead and can perform better
than either naive evaluation or asymptotically faster algorithms
for ranges of $n$ encountered in applications.
As an example,
fast numerical evaluation of the gamma function is investigated.
Our work generalizes two previous algorithms of Smith.
\end{abstract}




\overfullrule=5pt

\section{Introduction}

A sequence $(c(i))_{i=0}^{\infty}$ is called \emph{holonomic}
(or \emph{P-finite}) of order $r$ if it satisfies a linear recurrence equation
\begin{equation}
a_r(i) c(i+r) + a_{r-1}(i) c(i+r-1) + \ldots + a_0(i) c(i) = 0
\label{eq:holeq}
\end{equation}
where $a_0, \ldots, a_r$ are polynomials.
The class of holonomic sequences enjoys many useful closure properties:
for example, holonomic sequences form a ring, and if $c(i)$ is holonomic
then so is the sequence of partial sums $s(n) = \sum_{i=0}^n c(i).$
A~sequence is called \emph{hypergeometric} if it is holonomic of
order $r = 1$. The sequence of partial sums of a hypergeometric sequence
is holonomic of  order (at most) $r = 2$.

Many integer and polynomial sequences of interest in number theory and
combinatorics are holonomic, and the power series expansions of many well-known
special functions (such as the error function and Bessel functions)
are holonomic.

We are interested in efficient algorithms for evaluating an isolated term
$c(n)$ in a holonomic sequence when $n$ is large.
Section \ref{sect:matrix} recalls the well-known
techniques of rewriting \eqref{eq:holeq} in matrix form and applying
\emph{binary splitting}, which gives a near-optimal asymptotic speedup
for certain types of coefficients,
or \emph{fast multipoint evaluation} which
in the general case is the asymptotically fastest known algorithm.

In section \ref{sect:rectangular}, we give an algorithm
(Algorithm~\ref{alg:rsshift}) which becomes
efficient when the recurrence equation involves an ``expensive''
parameter (in a sense which is made precise),
based on the baby-step giant-step technique
of Paterson and Stockmeyer \cite{PatersonStockmeyer1973}
(called \emph{rectangular splitting} in \cite{mca}).

Our algorithm can be viewed as a generalization of the method given by Smith in
\cite{Smith2001} for computing rising factorials. Conceptually, it also
generalizes an algorithm given by Smith in \cite{Smith1989} for evaluation
of hypergeometric series. Our contribution is to
recognize that rectangular splitting can be applied
systematically to a very general class of sequences,
and in an efficient way (we provide a detailed cost analysis,
noting that some care is required in the construction
of the algorithm to get optimal performance).

The main intended application of rectangular splitting is
high-precision numerical evaluation
of special functions, where the parameter is a real or complex number
(represented by a floating-point approximation),
as discussed further in section \ref{sect:numerical}.
Although rectangular splitting is asymptotically slower than
fast multipoint evaluation, it is competitive in practice.
In section \ref{sect:casestudy}, we present implementation
results comparing several different algorithms for
numerical evaluation of the gamma function
to very high precision.

\section{Matrix algorithms} \label{sect:matrix}

Let $R$ be a commutative ring with unity and of sufficiently
large characteristic where necessary. 
Consider a sequence of rank-$r$ vectors $(c(i) = (c_1(i), \ldots, c_r(i))^T)_{i=0}^{\infty}$
satisfying a recurrence equation of the form
\begin{equation}
\begin{pmatrix}
c_1(i+1) \\ \vdots \\ c_r(i+1)
\end{pmatrix}
= M(i)
\begin{pmatrix}
c_1(i) \\ \vdots \\ c_r(i)
\end{pmatrix}
\label{eq:matrec}
\end{equation}
where $M \in R[k]^{r \times r}$ (or $\operatorname{Quot}(R)(k)^{r \times r}$)
and where $M(i)$ denotes entrywise
evaluation. Given an initial vector $c(0)$, we wish to evaluate
the single vector $c(n)$ for some $n > 0$, where we assume that no
denominator in $M$ vanishes for $0 \le i < n$.
A scalar recurrence of the form \eqref{eq:holeq}
can be rewritten as \eqref{eq:matrec} by taking the vector
to be $\tilde c(i) = (c(i), \ldots, c(i+r-1))^T$ and setting $M$
to the companion matrix
\begin{equation}
M = \frac{1}{a_r}
\begin{pmatrix}
       &      a_r &   &   \\
       &    & \ddots &    \\
       &      &         & a_r \\
-a_0   & -a_1 & \ldots & -a_{r-1}
\end{pmatrix}.
\label{eq:companion}
\end{equation}
%
In either case, we call the sequence holonomic (of order $r$).

Through repeated application of the recurrence equation,
$c(n)$ can be evaluated using $O(r^2 n)$ arithmetic operations
(or $O(r n)$ if $M$ is companion) and temporary storage of $O(r)$ values.
We call this strategy the \emph{naive algorithm}.

The naive algorithm is not generally
optimal for large $n$. The idea
behind faster algorithms is to first compute the matrix product
\begin{equation}
P(0,n) = \prod_{i=0}^{n-1} M(i).
\label{eq:matrixprod}
\end{equation}
and then multiply it by the vector of initial values
(matrix multiplication is of course noncommutative,
and throughout this paper the notation in \eqref{eq:matrixprod} is understood to mean
$M(n-1) \ldots M(2) M(1) M(0)$).
This increases the cost to $O(r^{\omega} n)$
arithmetic operations where $\omega$ is the exponent of matrix
multiplication, but we can save time for large $n$ by
exploiting the structure of the matrix product.
The improvement is most dramatic when all matrix entries are constant,
allowing binary exponentiation (with $O(\log n)$ complexity)
or diagonalization to be used,
although this is a rather special case.
We assume in the
remainder of this work that $r$ is fixed,
and omit $O(r^{\omega})$ factors from
any complexity estimates.

From this point, we may view the problem
as that of evaluating
\eqref{eq:matrixprod} for some $M \in R[k]^{r\times r}$.
It is not a restriction to demand
that the entries of $M$ are polynomials: if $M = \tilde M / q$,
we can write $P(0,n) = \tilde P(0,n) / Q(0,n)$ where
$\tilde P(0,n) = \prod_{i=0}^{n-1} q(i) \tilde M(i)$ and
$Q(0,n) = \prod_{i=0}^{n-1} q(i)$. This reduces the
problem to evaluating two denominator-free products,
where the second product has order 1.

\subsection{Binary splitting}

In the \emph{binary splitting} algorithm, we recursively compute a
product of square matrices
$P(a,b) = \prod_{i=a}^{b-1} M(i)$ (where the
entries of $M$ need not necessarily be polynomials of $i$), 
as
$P(m,b) P(a,m)$ where ${m = \lfloor (a+b) / 2 \rfloor}$.
If the entries of partial products grow in size,
this scheme balances the sizes of the subproducts in a way
that allows us to take advantage of fast multiplication.

For example, take $M(i) \in R[x]^{r\times r}$ where all $M(i)$
have bounded degree.
Then $P(a,b)$ has entries in $R[x]$ of
degree $O(b-a)$, and binary splitting can be shown to compute
$P(0,n)$ using $O(\M(n) \log n)$ operations in $R$ where
$\M(n)$ is the complexity of polynomial multiplication,
using $O(n)$ extra storage.
Over a general ring $R$, we have
$\M(n) = O(n \log^{1+o(1)} n)$ by the result of \cite{CantorKaltofen1991},
making the binary splitting softly optimal.
This is a significant improvement over the naive algorithm, which in general
uses $O(n^2)$ coefficient operations to generate the $n$-th entry in a
holonomic sequence of polynomials.

Analogously, binary splitting reduces
the bit complexity for evaluating holonomic sequences over $\mathbb{Z}$ or
$\mathbb{Q}$ (or more generally the algebraic numbers) from
$O(n^{2+o(1)})$ to $O(n^{1+o(1)})$. For further references and
several applications of the binary splitting technique,
we refer to Bernstein \cite{Bernstein2008}.

\subsection{Fast multipoint evaluation}

The \emph{fast multipoint evaluation} method is useful when all
arithmetic operations are assumed to have uniform cost. Fast
multipoint evaluation allows evaluating a polynomial of
degree $d$ simultaneously at $d$ points using $O(\textsf{M}(d) \log d)$
operations and $O(d \log d)$ space.
Applied to a polynomial matrix product, we obtain
Algorithm \ref{alg:multipoint}, which is due to
Chudnovsky and Chudnovsky \cite{ChudnovskyChudnovsky1988}.

\begin{algorithm}
  \caption{Polynomial matrix product using fast multipoint evaluation}
  \label{alg:multipoint}
  \begin{algorithmic}[1]
    \Require $M \in R[k]^{r \times r}$, $n = m \times w$
    \Ensure $\prod_{i=0}^{n-1} M(i)$
    \State $[T_0, \ldots, T_{m-1}] \gets [M(k), \ldots, M(k+m-1)]$
    \Statex \Comment{Compute entrywise Taylor shifts of the matrix}
    \State $U \gets \prod_{i=0}^{m-1} T_i$ \Comment{Binary splitting in $R[k]^{r \times r}$}
    \State $[V_0, \ldots, V_{w-1}] \gets [U(0), \ldots, U((w-1)m)]$
    \Statex \Comment{Fast multipoint evaluation}
    \State \Return{$\prod_{i=0}^{w-1} V_i$} \Comment{Repeated multiplication in $R^{r \times r}$}
  \end{algorithmic}
\end{algorithm}

We assume for simplicity of presentation
that $n$ is a multiple of the parameter
$m$ (in general, we can take $w = \lfloor n / m \rfloor$
and multiply by the remaining factors naively).
Taking $m \sim n^{1/2}$,
Algorithm \ref{alg:multipoint} requires $O(\M(n^{1/2}) \log n)$
arithmetic operations in the ring $R$,
using $O(n^{1/2} \log n)$ temporary storage
during the fast multipoint evaluation step.
Bostan, Gaudry and Schost \cite{BostanGaudrySchost2007} improve
the algorithm to obtain an $O(\textsf{M}(n^{1/2}))$
operation bound, which is the best available result for evaluating the
$n$-th term of a holonomic sequence over a general ring.
Algorithm \ref{alg:multipoint} and some of its applications
are studied further by Ziegler \cite{Ziegler2005}.

\section{Rectangular splitting for \\ parametric sequences} \label{sect:rectangular}

We now consider holonomic sequences whose recurrence equation
involves coefficients from a commutative ring $C$ with unity
as well as an additional, distinguished parameter $x$.
The setting is as in the previous section, but with
$R = C[x]$. In other words, we are considering
holonomic sequences of polynomials (or, by clearing denominators, rational functions)
of the parameter. We make the following definition.

\begin{definition}
A holonomic sequence
$(c(n) \equiv c(x,n))_{n=0}^{\infty}$
is parametric over $C$ (with parameter $x$)
if it satisfies a linear recurrence equation of
the form \eqref{eq:matrec} with $M \in R[k]$ where $R = C[x]$.
\end{definition}

Let $H$ be a commutative $C$-algebra, and let $c(x,k)$ be a
parametric holonomic sequence defined by a recurrence matrix
$M \in C[x][k]^{r\times r}$ and an initial vector
$c(z,0) \in H^r$.
Given some $z \in H$ and $n \in \mathbb{N}$, we wish to compute
the single vector $c(z,n) \in H^r$ efficiently subject to
the assumption that operations in $H$ are expensive
compared to operations in $C$. Accordingly, we distinguish between:
\begin{itemize}
\item \emph{Coefficient operations} in $C$
\item \emph{Scalar operations} in $H$ (additions in $H$ and multiplications $C \times H \to H$)
\item \emph{Nonscalar multiplications} $H \times H \to H$
\end{itemize}

For example, the sequence of rising factorials
$$c(x,n) = x^{\overline{n}} = x (x+1) \cdots (x+n-1)$$
is first-order holonomic (hypergeometric) with
the defining recurrence equation
$c(x,n+1) = (n+x) c(x,n),$
and parametric over $C = \mathbb{Z}$. In some applications,
we wish to evaluate $c(z,n)$ for $z \in H$ where
$H = \mathbb{R}$ or $H = \mathbb{C}$.

The Paterson-Stockmeyer algorithm \cite{PatersonStockmeyer1973}
solves the problem of evaluating a polynomial
$P(x) = \sum_{i=0}^{n-1} p_i x^i$ with $p_i \in C$
at $x = z \in H$ using a reduced
number of nonscalar multiplications. The idea is to write
the polynomial as a rectangular array
\begin{equation}
\begin{aligned}
P(x) & = (p_0 + \ldots + p_{m-1} x^{m-1}) \\
     & + (p_m + \ldots + p_{2m-1} x^{m-1}) x^m \\
     & + (p_{2m} + \ldots + p_{3m-1} x^{m-1}) x^{2m} \\
     & + \ldots
\end{aligned}
\label{eq:ps}
\end{equation}
After computing a table containing $x^2, x^3, \ldots, x^{m-1}$, the
inner (rowwise) evaluations can be done using only scalar multiplications,
and the outer (columnwise) evaluation with respect to $x^m$
can be done using about $n/m$ nonscalar multiplications.
With $m \sim n^{1/2}$, this algorithm requires $O(n^{1/2})$ nonscalar
multiplications and $O(n)$ scalar operations.

A straightforward application of the Paterson-Stockmeyer
algorithm to evaluate each entry of $\prod_{i=0}^{n-1} M(x,i) \in C[x]^{r \times r}$
yields Algorithm \ref{alg:rsps}
and the corresponding complexity estimate of Theorem \ref{propsqrt}.

\begin{algorithm}
  \caption{Polynomial matrix product and evaluation using rectangular splitting}
  \label{alg:rsps}
  \begin{algorithmic}[1]
    \Require $M \in C[x][k]^{r \times r}$, $z \in H$, $n = m \times w$
    \Ensure $\prod_{i=0}^{n-1} M(z,i) \in H^{r \times r}$
    \State $[T_0, \ldots, T_{n-1}] \gets [M(x,0), \ldots, M(x,n-1)]$
    \Statex \Comment{Evaluate matrix w.r.t. $k$, giving $T_i \in C[x]^{r \times r}$}
    \State $U \gets \prod_{i=0}^{n-1} T_i$ \Comment{Binary splitting in $C[x]^{r\times r}$}
    \State $V \gets U(z)$
    \Statex \Comment{Evaluate $U$ entrywise using Paterson-Stockmeyer with step length $m$}
    \State \Return{$V$}
  \end{algorithmic}
\end{algorithm}

\begin{theorem}
The $n$-th entry in a parametric holonomic sequence can be evaluated
using $\OO(n^{1/2})$ nonscalar multiplications, $\OO(n)$ scalar
operations, and $\OO(\M(n) \log n)$ coefficient operations.
\label{propsqrt}
\end{theorem}

\begin{proof}
We call Algorithm \ref{alg:rsps} with $m \sim n^{1/2}$.
Letting $d = \max \deg_k M$ and $e = \max \deg_x M$,
computing $T_0, \ldots, T_{n-1}$
takes $O(nde) = O(n)$ coefficient operations.
Since $\deg_x T_i \le e$, generating $U$
using binary splitting costs
$O(\M(n) \log n)$ coefficient operations.
Each entry in $U$ has degree at most $ne = O(n)$,
and can thus be evaluated using $O(n^{1/2})$
nonscalar multiplications and $O(n)$ scalar operations
with the Paterson-Stockmeyer algorithm.
\end{proof}

If we only count nonscalar multiplications, Theorem \ref{propsqrt}
is an asymptotic improvement over fast multipoint evaluation which
uses $O(n^{1/2} \log^{2+o(1)} n)$ nonscalar multiplications
($O(n^{1/2} \log^{1+o(1)} n)$ with the improvement of
Bostan, Gaudry and Schost).

Algorithm \ref{alg:rsps} is not ideal in practice
since the polynomials in $U$ grow to degree $O(n)$.
Their coefficients also grow
to $O(n \log n)$ bits when $C = \mathbb{Z}$ (for example,
in the case of rising factorials,
the coefficients are the Stirling numbers of the first
kind $S(n,k)$ which grow to a magnitude between $(n-1)!$ and $n!$).
This problem can be mitigated by repeatedly applying Algorithm \ref{alg:rsps}
to successive subproducts $\prod_{i=a}^{a+\tilde n} M(z,i)$ where $\tilde n \ll n$,
but the nonscalar complexity is then no longer the best possible.
A better strategy is to apply
rectangular splitting to the matrix product itself, leading to
Algorithm~\ref{alg:rsshift}.
We can then reach the same operation complexity
while only working with polynomials of degree $O(n^{1/2})$,
and over $C = \mathbb{Z}$, having coefficients of bit size $O(n^{1/2} \log n)$.

\begin{theorem}
For any choice of $m$, Algorithm~\ref{alg:rsshift} requires
$O(m + n/m)$ nonscalar multiplications,
$O(n)$ scalar operations, and $O((n/m) \M(m) \log m)$ coefficient operations.
In particular, the complexity bounds stated in Theorem \ref{propsqrt}
also hold for Algorithm \ref{alg:rsshift} with $m \sim n^{1/2}$.
Moreover, Algorithm \ref{alg:rsshift} only
requires storage of $O(m)$ elements of $C$ and~$H$,
and if $C = \mathbb{Z}$, the coefficients have bit size $O(m \log m)$.
\end{theorem}

\begin{proof}
This follows by applying a similar argument as used in the proof of
Theorem~\ref{propsqrt}
to the operations in the inner loop of Algorithm~\ref{alg:rsshift},
noting that $U$ has entries of degree $m \deg_x M = O(m)$
and that the matrix multiplication $S \times V$
requires $O(1)$ nonscalar multiplications and scalar
operations (recalling that we consider $r$ fixed).
\end{proof}

\begin{algorithm}
  \caption{Improved polynomial matrix product and evaluation using rectangular splitting}
  \label{alg:rsshift}
  \begin{algorithmic}[1]
    \Require $M \in C[x][k]^{r \times r}$, $z \in H$, $n = m \times w$
    \Ensure $\prod_{i=0}^{n-1} M(z,i) \in H^{r \times r}$
    \State Compute power table $[z^j$, $0 \le j \le m \deg_x M]$
    \State $V \gets 1_{H^{r\times r}}$ \Comment{Start with the identity matrix}
    \For{$i \gets 0 \textrm{ to } w - 1$}
        \State $[T_0, \ldots, T_{m-1}] \gets [M(x,im+j)]_{j=0}^{m-1}$
        \Statex \Comment{Evaluate matrix w.r.t. $k$, giving $T_j \in C[x]^{r \times r}$}
        \State $U \gets \prod_{j=0}^{m-1} T_j$ \Comment{Binary splitting in $C[x]^{r\times r}$}
        \State $S \gets U(z)$ \Comment{Evaluate w.r.t. $x$ using power table}
        \State $V \gets S \times V$ \Comment{Multiplication in $H^{r\times r}$}
    \EndFor
    \State \Return{$V$}
  \end{algorithmic}
\end{algorithm}

\subsection{Variations}

Many variations of Algorithm~\ref{alg:rsshift}
are possible. Instead of using binary
splitting directly to compute $U$, we can generate the bivariate matrix
\begin{equation}
W_m = \prod_{i=0}^{m-1} M(x,k+i) \in C[x][k]^{r \times r}
\end{equation}
at the start of the algorithm, and then
obtain $U$ by evaluating $W_m$ at $k = im$.
We may also work with differences of two successive $U$
(for small $m$, this can introduce cancellation
resulting in slightly smaller polynomials or coefficients).
Combining both variations, we end up with Algorithm~\ref{alg:rectdelta}
in which we expand and evaluate the bivariate polynomial matrices
\begin{equation*}
\Delta_m = \prod_{i=0}^{m-1} M(x,k+m+i) - \prod_{i=0}^{m-1} M(x,k+i) \in C[x][k]^{r \times r}.
\end{equation*}
This version of the rectangular splitting algorithm can be viewed as a generalization of an
algorithm used by Smith \cite{Smith2001} for computing rising factorials (we consider
the case of rising factorials further in Section \ref{sect:rising}).
In fact, the author of the present paper first
found Algorithm~\ref{alg:rectdelta} by generalizing
Smith's algorithm, and only later discovered
Algorithm \ref{alg:rsshift} by ``interpolation'' between
Algorithm~\ref{alg:rsps} and Algorithm~\ref{alg:rectdelta}.

\begin{algorithm}
  \caption{Polynomial matrix product and evaluation using rectangular splitting (variation)}
  \label{alg:rectdelta}
  \begin{algorithmic}[1]
    \Require $M \in C[x][k]^{r \times r}$, $z \in H$, $n = m \times w$
    \Ensure $\prod_{i=0}^{n-1} M(z,i) \in H^{r \times r}$
    \State Compute power table $[z^j$, $0 \le j \le m \deg_x M]$
    \State $\Delta \gets \prod_{i=0}^{m-1} M(x,k+m+i) - \prod_{i=0}^{m-1} M(x,k+i)$
    \Statex \Comment{Binary splitting in $C[x][k]^{r \times r}$}
    \State $V \gets S \gets \prod_{i=0}^{m-1} M(z,i)$
    \Statex \Comment{Evaluate w.r.t. $k$, and w.r.t. $x$ using power table}
    \For{$i \gets 0 \textrm{ to } w - 2$}
      \State $S \gets S + \Delta(z,mi)$
      \Statex \Comment{Evaluate w.r.t. $k$, and w.r.t. $x$ using power table}
      \State $V \gets S \times V$
    \EndFor
    \State \Return{$V$}
  \end{algorithmic}
\end{algorithm}

The efficiency of Algorithm~\ref{alg:rectdelta} is theoretically
somewhat worse than that of Algorithm~\ref{alg:rsshift}.
Since $\deg_x W_m = O(m)$ and $\deg_k W_m = O(m)$,
$W_m$ has $O(m^2)$ terms (likewise for $\Delta_m$),
making the space complexity higher and increasing
the number of coefficient operations to $O((n/m) m^2)$
for the evaluations with respect to $k$.
However, this added cost may be negligible in practice.
Crucially, when $C = \mathbb{Z}$, the coefficients
have similar bit sizes as in Algorithm \ref{alg:rsshift}.

Initially generating $W_m$ or $\Delta_m$ also adds some cost,
but this is cheap compared to the evaluations when $n$ is
large enough: binary splitting over $C[x][k]$
costs $O(\M(m^2) \log m)$ coefficient operations by
Lemma 8.2 and Corollary 8.28 in \cite{vonzurGathenGerhard2003}.
This is essentially the same as the total cost of binary
splitting in Algorithm~\ref{alg:rsshift} when $m \sim n^{1/2}$.

We also note that a small improvement to Algorithm~\ref{alg:rsshift}
is possible if $M(x,k+m) = M(x+m,k)$:
instead of computing $U$ from scratch using binary splitting
in each loop iteration, we can update it using a Taylor shift.
At least in sufficiently large characteristic,
the Taylor shift can be computed using $O(\M(m))$ coefficient
operations with the
convolution algorithm of Aho, Steiglitz and Ullman \cite{Aho1975evaluating},
saving a factor $O(\log n)$ in the total number of coefficient operations.
In practice,
basecase Taylor shift algorithms may also be beneficial (see \cite{von1997fast}).

In lucky cases, the polynomial coefficients (in either
Algorithm \ref{alg:rsshift} or \ref{alg:rectdelta}) might
satisfy a recurrence relation, allowing them
to be generated using $O(n)$ coefficient operations
(and avoiding the dependency on polynomial arithmetic).


\subsection{Several parameters}

The rectangular splitting technique can be generalized
to sequences $c(x_1,\ldots,x_v,k)$ depending on several parameters.
In Algorithm~\ref{alg:rsshift}, we simply replace the power
table by a $v$-dimensional array of the possible
monomial combinations. Then we 
have the following result (ignoring coefficient operations).

\begin{theorem}
The $n$-th entry in a holonomic sequence depending on $v$ parameters
can be evaluated with rectangular splitting
using $O(m^v + n/m)$ nonscalar multiplications and
$O((n/m) m^v)$ scalar multiplications.
In particular, taking $m = n^{1/(v+1)}$, $O(n^{1-1/v})$
nonscalar multiplications and $O(n^{2v / (1+v)})$ scalar multiplications
suffice.
\label{thm:several}
\end{theorem}

\begin{proof}
If $d_i = \operatorname{deg}_{x_i} M \le d$, the entries of a
product of $m$ successive shifts of $M$ are $C$-linear
combinations of $x_1^{e_{1,j}} \cdots x_h^{e_{v,j}}$,
$0 \le e_{i,j} \le m d_i \le md$, so there is a total of $O(m^v)$ powers.
\end{proof}

Unfortunately, this gives rapidly diminishing returns for large $v$.
When $v > 1$, the number of nonscalar multiplications
according to Theorem~\ref{thm:several} is
asymptotically worse than with fast multipoint evaluation,
and reducing the number of nonscalar multiplication
requires us to perform more than $O(n)$ scalar multiplications,
as shown in Table \ref{tab:dimension}.
Nevertheless, rectangular splitting could perhaps still be useful
in some settings where the cost of nonscalar
multiplications is sufficiently large.

\begin{table}[ht!]
\centering
\begin{tabular}{ l | l l l }
$v$ & $m$ & Nonscalar & Scalar\\ \hline
$1$ & $n^{1/2}$ & $O(n^{0.5})$ & $O(n)$ \\
$2$ & $n^{1/3}$ & $O(n^{0.666\ldots})$ & $O(n^{1.333\ldots})$ \\
$3$ & $n^{1/4}$ & $O(n^{0.75})$ & $O(n^{1.5})$ \\
$4$ & $n^{1/5}$ & $O(n^{0.8})$ & $O(n^{1.6})$ \\
\end{tabular}
\caption{Step size $m$ minimizing the number of
nonscalar multiplications for rectangular splitting involving $v$ parameters.}
\label{tab:dimension}
\end{table}

\section{Numerical evaluation} \label{sect:numerical}

Assume that we want to evaluate $c(x,n)$
where the underlying coefficient ring is
$C = \mathbb{Z}$ (or $\overline{\mathbb{Q}}$)
and the parameter $x$ is a real or complex number
represented by a floating-point approximation with
a precision of $p$ bits.

Let $\mathsf{M}_{\mathbb{Z}}(p)$ denote the bit complexity of some algorithm
for multiplying two $p$-bit integers or floating-point numbers. Commonly
used algorithms include classical multiplication with
$\mathsf{M}_{\mathbb{Z}}(p) = O(p^2)$, Karatsuba multiplication with
$\mathsf{M}_{\mathbb{Z}}(p) = O(p^{1.585})$, and
Fast Fourier Transform (FFT) based multiplication,
such as the Sch\"{o}nhage-Strassen algorithm, with
$\mathsf{M}_{\mathbb{Z}}(p) =  \OOsoft(p)$.
An unbalanced multiplication where the smaller operand has
$q$ bits can be done using
$O((p/q) \mathsf{M}_{\mathbb{Z}}(q))$ bit operations \cite{mca}.

The naive algorithm clearly uses $O(n \MZ(p))$ bit operations
to evaluate $c(x,n)$, or $\OOsoft(np)$ with FFT multiplication.
In Algorithm~\ref{alg:rsshift}, the nonscalar
multiplications cost $O((m + n/m) \MZ(p))$ bit operations.
The coefficient operations cost $\tilde O(mn)$
bit operations (assuming the use of fast polynomial arithmetic),
which becomes negligible if $p$ grows faster than $m$.
Finally, the scalar multiplications (which are unbalanced) cost
\begin{equation*}
O\left(n \, p \, \frac{\MZ(m \log m)}{m \log m} \right)
\end{equation*}
bit operations.
Taking $m \sim n^{\alpha}$ for $0 < \alpha < 1$, we get
an asymptotic speedup with classical or Karatsuba multiplication
(see Table~\ref{tab:multalg})
provided that $p$ grows sufficiently rapidly along with $n$.
With FFT multiplication, the scalar multiplications
become as expensive as the nonscalar multiplications,
and rectangular therefore does not give an asymptotic
improvement.

\begin{table}[ht!]
\centering
\begin{tabular}{ l | l l }
Mult. algorithm & Scalar multiplications & Naive \\ \hline
Classical & $\OOsoft(n^{1+\alpha} p)$   & $\OOsoft(np^2)$ \\
Karatsuba & $\OOsoft(n^{1+0.585\alpha} p)$ & $\OOsoft(np^{1.585})$ \\
FFT       & $\OOsoft(np)$           & $\OOsoft(np)$ \\
\end{tabular}
\caption{Bit complexity of scalar multiplications in Algorithm~\ref{alg:rsshift}
and total bit complexity of the naive algorithm}
\label{tab:multalg}
\end{table}

However, due to the overhead of FFT multiplication,
rectangular splitting is still likely to save a constant
factor over the naive algorithm. In practice,
one does not necessarily get the best performance by choosing
$m \approx n^{0.5}$ to minimize the number of nonscalar multiplications alone;
the best $m$ has to be determined empirically.

Algorithm~\ref{alg:multipoint} is asymptotically
faster than the naive algorithm as well
as rectangular splitting, with a
bit complexity of $\OOsoft(n^{1/2} p)$.
It should be noted that this estimate does not reflect the complexity
required to obtain a given \emph{accuracy}.
As observed by K\"{o}hler and Ziegler \cite{KohlerZiegler2008},
fast multipoint evaluation can exhibit
poor numerical stability,
suggesting that~$p$ might have to
grow at least as fast as $n$ to get accuracy proportional to $p$.

When $x$ and all coefficients in $M$ are
positive, rectangular splitting introduces no subtractions that
can cause catastrophic cancellation, and the reduction
of nonscalar multiplications
even improves stability compared to the naive algorithm,
making $O(\log n)$ guard bits sufficient to reach $p$-bit accuracy.
When sign changes are present, evaluating degree-$m$
polynomials in expanded form can reduce accuracy, typically
requiring use of $\tilde O(m)$ guard bits. In this case
Algorithm~\ref{alg:rsshift} is a
marked improvement over Algorithm~\ref{alg:rsps}.

\subsection{Summation of power series}

A common situation is that we wish to evaluate
a truncated power series
\begin{equation}
f(x) \approx s(x,n) = \sum_{k=0}^n c(k) x^k, \quad n = O(p)
\label{eq:fseries}
\end{equation}
where $c(k)$ is a holonomic sequence taking rational (or algebraic)
values and $x$ is a real or complex number.
In this case the Paterson-Stockmeyer algorithm is applicable,
but might not give a speedup when applied directly
as in Algorithm~\ref{alg:rsps} due to the growth of the coefficients.
Since $d(k) = c(k) x^k$ and $s(x,n)$ are
holonomic sequences with $x$ as parameter,
Algorithm~\ref{alg:rsshift} is applicable.

Smith noted in \cite{Smith1989} that when $c(k)$ is hypergeometric
(Smith considered the Taylor expansions of elementary functions
in particular),
the Paterson-Stockmeyer technique can be combined with scalar divisions
to remove accumulated factors from the coefficients.
This keeps all scalars at a size of $O(\log n)$ bits,
giving a speedup over naive evaluation when
non-FFT multiplication is used
(and when scalar divisions are assumed to be roughly
as cheap as scalar multiplications).
This algorithm is studied in more detail by Brent and Zimmermann \cite{mca}.

At least conceptually, Algorithm~\ref{alg:rsshift} can be viewed as a
generalization of Smith's hypergeometric summation algorithm
to arbitrary holonomic sequences depending on a parameter
(both algorithms can be viewed as means to
eliminate repeated content from the associated matrix product).
The speedup is not quite as good since we only reduce
the coefficients to $O(n^{1/2} \log n)$ bits versus Smith's $O(\log n)$.
However, even for hypergeometric series, Algorithm~\ref{alg:rsshift}
can be slightly faster than Smith's algorithm for small $n$
(e.g. $n \lesssim 100$) since divisions tend to be more expensive than
scalar multiplications in implementations.

Algorithm \ref{alg:rsshift} is also more general:
for example, we can use it to evaluate the
generalized hypergeometric function
\begin{equation}
\bigpFq{p}{q}{a_1, \ldots, a_p}{b_1, \ldots, b_q}{\!z} =
\sum_{k=0}^\infty
\frac{a_1^{\overline k} \cdots a_p^{\overline k}}{b_1^{\overline k} \cdots b_q^{\overline k}} \, \frac {w^k}{k!}
\end{equation}
where $a_i, b_i, w$ (as opposed to $w$ alone) are rational functions of
the real or complex parameter $x$.

An interesting question, which we do not attempt to answer here, is
whether there is a larger class of parametric sequences other than
hypergeometric sequences and their sums for which we can reduce
the number of nonscalar multiplications to $O(n^{1/2})$
while working with coefficients that are strictly smaller than
$O(n^{1/2} \log n)$ bits.

\subsection{Comparison with asymptotically faster \\ algorithms}

If all coefficients in \eqref{eq:fseries} including the parameter $x$
are rational or algebraic numbers and the series converges, $f(x)$ can be evaluated
to $p$-bit precision using $\tilde O(p)$ bit operations
using binary splitting. 
An $\tilde O(p)$ bit complexity can also
be achieved for arbitrary real or complex $x$ by combining
binary splitting with translation of the
differential equation for $f(x)$.
The general version of this algorithm, sometimes called the \emph{bit-burst algorithm},
was developed by Chudnovsky and Chudnovsky and independently
with improvements by van der Hoeven \cite{vdH:hol}. It is used
in some form
for evaluating elementary functions in several libraries,
and a general version has been implemented by Mezzarobba \cite{Mezzarobba2010}.

For high-precision evaluation of elementary functions, binary splitting typically only
becomes worthwhile at a precision of several thousand digits,
while implementations typically use Smith's algorithm for summation of hypergeometric series
at lower precision.
We expect that Algorithm~\ref{alg:rsshift} can be used
in a similar fashion for a wider class of special functions.

When $c(k)$ in $\eqref{eq:fseries}$ involves real or complex numbers,
binary splitting no longer gives a speedup. In this case,
we can use fast multipoint to evaluate \eqref{eq:fseries} using
$\OOsoft(p^{1.5})$ bit operations
(Borwein \cite{Borwein1987} discusses the application to numerical evaluation
of hypergeometric functions). This method does not appear to be
widely used in practice, presumably owing to its high overhead and relative
implementation difficulty. Although rectangular splitting is
not as fast asymptotically, its ease of implementation and low overhead
makes it an attractive alternative.

\section{High-precision computation of the gamma function} \label{sect:casestudy}

In this section, we consider two holonomic sequences depending
on a numerical parameter: rising factorials, and the partial
sums of a certain hypergeometric series defining the incomplete
gamma function. In both cases, our goal is to accelerate numerical
evaluation of the gamma function at very high precision.

We have implemented the algorithms in the present section using
floating-point ball arithmetic (with rigorous error bounding)
as part of the Arb library\footnote{\url{http://fredrikj.net/arb}}.
All arithmetic in $\mathbb{Z}[x]$ is done via FLINT \cite{Hart2010},
using a Sch\"{o}nhage-Strassen FFT implemented by W. Hart.

Fast numerically stable multiplication in $\mathbb{R}[x]$ is
done by breaking polynomials into segments with
similarly-sized coefficients and computing the subproducts
exactly in $\mathbb{Z}[x]$ (a simplified version of
van der Hoeven's block multiplication algorithm
\cite{vdH:stablemult}), and asymptotically fast
polynomial division is implemented using Newton iteration.

All benchmark results were obtained on a 2.0 GHz Intel Xeon E5-2650 CPU.

\subsection{Rising factorials} \label{sect:rising}

Rising factorials of a real or complex argument appear when
evaluating the gamma function via the asymptotic Stirling series
\begin{align*}
\log \Gamma(x) & = \left(x-\frac{1}{2}\right) \log x - x + \frac{\log 2 \pi}{2} \\
& + \sum_{k=1}^{N-1} \frac{B_{2k}}{2k(2k-1)x^{2k-1}} + R_N(x).
\end{align*}
To compute $\Gamma(x)$ with $p$-bit accuracy, we choose a
positive integer $n$ such that there is an $N$ for which
$|R_N(x+n)| < 2^{-p}$, and then evaluate
${\Gamma(x) = \Gamma(x+n) / x^{\overline{n}}}$.
It is sufficient to choose $n$ such that the
real part of $x + n$ is of order $\beta p$ where
$\beta = (2 \pi)^{-1} \log 2 \approx 0.11$.

The efficiency of the Stirling series can be improved
by choosing $n$ slightly larger than the absolute minimum
in order to reduce $N$. For example,
$\mathrm{Re}(x+n) \approx 2 \beta p$ is a good choice.
A faster rising factorial is doubly advantageous: it
speeds up the argument reduction, and making larger $n$ cheap
allows us to get away with fewer Bernoulli numbers.

Smith \cite{Smith2001} uses the difference of four consecutive terms
\begin{align*}
(x+k+4)^{\overline{4}} - (x+k)^{\overline{4}} & = (840 + 632 k + 168 k^2 + 16 k^3) \\
          & + (632 + 336 k + 48 k^2) x \\
          & + (168 + 48 k) x^2 \\
          & + 16 x^3
\end{align*}
to reduce the number of nonscalar multiplications to compute $x^{\overline{n}}$ from $n-1$
to about $n / 4$. This is precisely Algorithm~\ref{alg:rectdelta}
specialized to the sequence of rising factorials and with a fixed step length $m = 4$.

Consider Smith's algorithm with a variable step length~$m$.
Using the binomial theorem and some rearrangements,
the polynomials can be written down explicitly as
\begin{equation}
\Delta_m = (x+k+m)^{\overline{m}} - (x+k)^{\overline{m}}
 = \sum_{v=0}^{m-1} x^v \sum_{i=0}^{m-v-1} k^i \; C_m(v,i)
\end{equation}
where
\begin{equation}
C_m(v,i) = \sum_{j=i+1}^{m-v} m^{j-i} S(m,v+j) {{v+j} \choose v} {j \choose i}
\label{eq:ccoeff}
\end{equation}
and where $S(m,v+j)$ denotes an unsigned Stirling number of the first kind.
This formula can be used to generate $\Delta_m$ efficiently in practice without
requiring bivariate polynomial arithmetic. In fact, the coefficients can
be generated even cheaper by taking advantage of the recurrence (found by
M. Kauers)
\begin{equation}
(v+1) C_m(v+1,i) = (i+1) C_m(v, i+1).
\end{equation}

We have implemented several algorithm for evaluating the rising
factorial of a real or complex number.
For tuning parameters, we empirically determined simple
formulas that give nearly optimal
performance for different combinations of $n, p < 10^5$
(typically within 20\% of the speed with the best tuning
value found by a brute force search):
\begin{itemize}
\item In Algorithm~\ref{alg:multipoint}, $m = n^{0.5}$.
\item Algorithm~\ref{alg:rsps} is applied to subproducts of length $\tilde n = \min(2 n^{0.5}, 10 p^{0.25})$, with $m = \tilde n^{0.5}$.
\item In Algorithms~\ref{alg:rsshift} and \ref{alg:rectdelta}, $m = \min(0.2 p^{0.4}, n^{0.5})$.
\end{itemize}

Our implementation of Algorithm~\ref{alg:rectdelta} uses \eqref{eq:ccoeff}
instead of binary splitting, and Algorithm~\ref{alg:rsshift}
exploits the symmetry of $x$ and $k$ to update the
matrix $U$ using Taylor shifts instead of repeated binary splitting.

Figure \ref{fig:rising} compares the running times where $x$
is a real number with a precision of $p = 4n$ bits. This input
corresponds to that used in our Stirling series implementation
of the gamma function.

\begin{figure}[width=8cm] \label{fig:rising}
\begin{center}
\includegraphics[width=8cm]{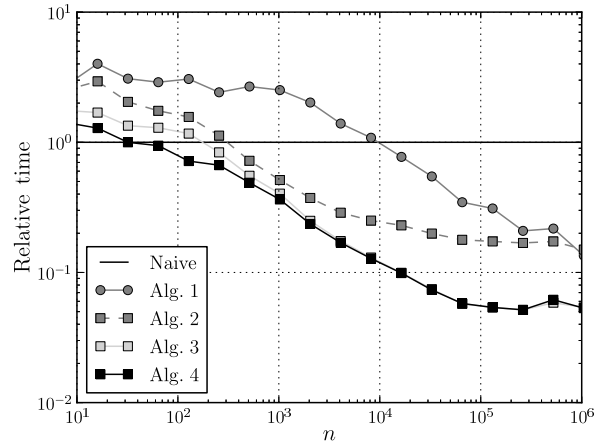}
\caption{Timings of rising factorial algorithms, normalized
against the naive algorithm.}
\end{center}
\end{figure}

On this benchmark, Algorithms \ref{alg:rsshift} and \ref{alg:rectdelta} are the
best by far, gaining a 20-fold speedup over the naive algorithm for large $n$
(the speedup levels off
around $n = 10^5$, which is expected since this is
the approximate point where FFT integer multiplication kicks in).
Algorithm~\ref{alg:rectdelta} is slightly faster than Algorithm \ref{alg:rsshift}
for $n < 10^3$, even narrowly beating the naive algorithm for
$n$ as small as $\approx 10^2$.

Algorithm~\ref{alg:multipoint} (fast multipoint evaluation)
has the most overhead of all algorithms and only overtakes
the naive algorithm around $n = 10^4$ (at a precision of $40,000$ bits).
Despite its theoretical advantage, it is slower than rectangular splitting
up to $n$ exceeding $10^6$.

Table \ref{tab:gammatime} shows absolute timings for evaluating $\Gamma(x)$
where $x$ is a small real number in Pari/GP 2.5.4, and our implementation
in Arb (we omit results for MPFR 3.1.1 and Mathematica 9.0, which
both were slower than Pari).
Both implementations use the Stirling series, caching
the Bernoulli numbers to speed up multiple evaluations.
The better speed of Arb for a repeated evaluation (where the Bernoulli numbers
are already cached) is mainly due to the use of rectangular splitting
to evaluate the rising factorial.
The total speedup is smaller than it would be for computing the
rising factorial alone since we still have to evaluate the Bernoulli
number sum in the Stirling series.
The gamma function implementations
over $\mathbb{C}$ have similar characteristics.

\begin{table}[ht!] \label{tab:gammatime}
\centering
\begin{tabular}{r | l l l l}
Decimals & Pari/GP & (first) & Arb & (first) \\ \hline
100 & 0.000088 & & 0.00010 & \\
300 & 0.00048 & & 0.00036 & \\
1000 & 0.0057 & & 0.0025 & \\
3000 & 0.072 & (9.2) & 0.021 & (0.090) \\
10000 & 1.2 & (324) & 0.25 & (1.4) \\
30000 & 15 & (8697) & 2.7 & (22) \\
100000 &   & &  39 & (433) \\
300000 &   & &  431 & (7131) \\
\end{tabular}
\caption{Timings in seconds for evaluating $\Gamma(x)$
where $x$ is a small real number (timings for the first evaluation,
including Bernoulli number generation, is shown in parentheses).}
\end{table}

\subsection{A one-parameter hypergeometric series}

The gamma function can be approximated via the (lower)
incomplete gamma function as
\begin{equation}
\Gamma(z) \approx \gamma(z,N) = z^{-1} N^z e^{-N} \,_1F_1(1,1+z,N).
\label{eq:gammahyper}
\end{equation}
Borwein \cite{Borwein1987} noted that applying fast multipoint evaluation
to a suitable truncation of the hypergeometric series in \eqref{eq:gammahyper}
allows evaluating the gamma function
of a fixed real or complex argument to $p$-bit precision
using $\OOsoft(p^{1.5})$ bit operations,
which is the best known result for general $z$ (if $z$ is algebraic,
binary splitting evaluation of the same series
achieves a complexity of $\OOsoft(p)$).

Let $t_k = N^k / (z (z+1) \cdots (z+k))$ and $s_n = \sum_{k=0}^n t_k$,
giving
$$\lim_{n\to\infty} s_n = \,_1F_1(1,1+z,N) / z.$$
For $z \in [1,2]$, choosing
$N \approx p \log 2$ and $n \approx (e \log 2) p$
gives an error of order $2^{-p}$ (it
is easy to compute strict bounds). The partial sums
satisfy the order-2 recurrence
\begin{equation}
\begin{pmatrix} s_k \\ t_{k+1} \end{pmatrix} =
\frac{M(k)}{q(k)} \frac{M(k-1)}{q(k-1)} \cdots \frac{M(0)}{q(0)}
\label{eq:gammahypproduct}
\begin{pmatrix} 0 \\ 1/z \end{pmatrix}
\end{equation}
where
\begin{equation}
M(k) =
\begin{pmatrix}
1+k+z & 1+k+z \\
0     & N     \\
\end{pmatrix}, \quad
q(k) = 1+k+z.
\end{equation}
The matrix product \eqref{eq:gammahypproduct} may be computed
using fast multipoint evaluation
or rectangular splitting. We note that the denominators
are identical to the top left entries of the numerator matrices,
and therefore do not need to be computed separately.

Figure \ref{fig:gammahyp} compares
the performance of the Stirling series (with fast argument reduction using rectangular splitting)
and three different implementations of the ${}_1F_1$ series (naive summation,
fast multipoint evaluation, and rectangular splitting
using Algorithm~\ref{alg:rsshift} with $m = 0.2 n^{0.4}$)
for evaluating $\Gamma(x)$ where $x$ is a real argument
close to unity.

\begin{figure}[width=8cm] \label{fig:gammahyp}
\begin{center}
\includegraphics[width=8cm]{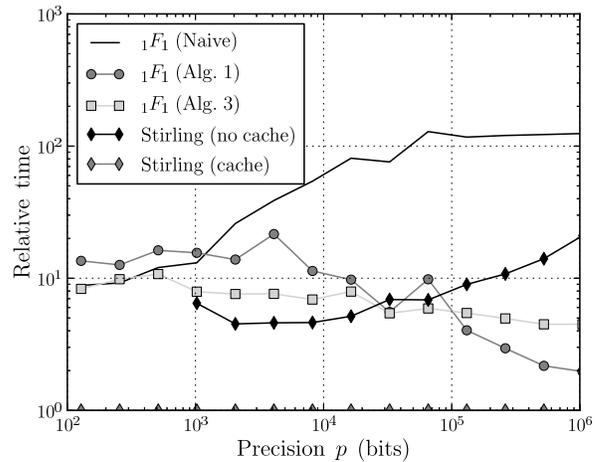}
\caption{Timings of gamma function algorithms,
normalized against the Stirling series with Bernoulli numbers cached.}
\end{center}
\end{figure}

Both fast multipoint evaluation and rectangular splitting
speed up the hypergeometric series
compared to naive summation.
Using either algorithm, the
hypergeometric series is competitive with the Stirling
series for a single evaluation
at precisions above roughly 10,000 decimal digits.

Algorithm~\ref{alg:multipoint} performs better than on
the rising factorial benchmark, and is
faster than Algorithm~\ref{alg:rsshift} above $10^5$ bits.
A possible explanation for this difference is that
the number of terms used in the hypergeometric series roughly is $n \approx 2p$
where $p$ is the precision in bits, compared to $n \approx p / 4$
for the rising factorial, and rectangular splitting
favors higher precision and fewer terms.

The speed of Algorithm~\ref{alg:rsshift} is remarkably close to
that of Algorithm~\ref{alg:multipoint} even for $p$
as large as $10^6$.
Despite being asymptotically slower, the
simplicity of rectangular splitting
combined with its lower memory consumption and better
numerical stability (in our implementation, Algorithm~\ref{alg:rsshift}
only loses a few significant digits,
while Algorithm~\ref{alg:multipoint} loses a few percent of the number of significant digits)
makes it an attractive option for extremely high-precision
evaluation of the gamma function.

Once the Bernoulli numbers have been cached after the first evaluation, the Stirling series
still has a clear advantage up to precisions exceeding $10^6$ bits.
We may remark that our implementation of the Stirling series
has been optimized for multiple evaluations: by choosing
larger rising factorials and generating the Bernoulli numbers dynamically
without storing them, both the speed and memory consumption
for a single evaluation could be improved.

\section{Discussion}

We have shown that rectangular splitting can be profitably
applied to evaluation of a general class of linearly recurrent sequences.
When used for numerical evaluation of special functions, our benchmark results
indicate that rectangular splitting can be faster than either naive evaluation
or fast multipoint evaluation over a wide precision range
(between approximately $10^3$ and $10^6$ bits).

Two natural questions are whether this approach can be generalized
further (to more general classes of sequences), and whether
it can be optimized further (perhaps for more specific classes
of sequences).

\section{Acknowledgements}

The author thanks Manuel Kauers for providing
useful feedback on draft versions of this paper.

\bibliographystyle{plain}
\bibliography{references.bib}

\end{document}